\documentclass[10pt,superscriptaddress,nofootinbib,notitlepage,twocolumn,tightenlines]{revtex4-1}

\usepackage[normalem]{ulem}
\usepackage{amsthm}
\usepackage{amsmath}
\usepackage{amssymb}
\usepackage{graphicx,epstopdf}
\usepackage{url}
\usepackage{dsfont}
\usepackage{float}
\usepackage{bm}
\usepackage{ifthen}
\usepackage[usenames,dvipsnames]{color}
\usepackage{mathrsfs}
\usepackage[colorlinks=true,citecolor=blue,urlcolor=black]{hyperref}
\usepackage{float}

\newcommand{\Tr}{\mathrm{ Tr }}
\newcommand{\tn}[1]{\textnormal{#1}}
\newcommand{\be}{\begin{equation}}
\newcommand{\ee}{\end{equation}}

\newcommand{\p}{\prime}
\newcommand{\leak}{\tn{leak}_{\tn{EC}}}
\newcommand{\epscorr}{\eps_{\tn{cor}}}
\newcommand{\Stol}{S_{\tn{tol}}}
\newcommand{\Qtol}{Q_{\tn{tol}}}
\newcommand{\Stes}{S_{\tn{test}}}
\newcommand{\Qtes}{Q_{\tn{test}}}
\newcommand{\esec}{\eps_{\tn{sec}}}
\newcommand{\devicesource}{\mathbb{S}}
\newcommand{\devicekey}{\mathbb{M}_{\tn{key}}}
\newcommand{\devicetest}{\mathbb{M}_{\tn{test}}}

\newcommand{\sket}[1]{{\ensuremath{\lvert#1\rangle}}}
\newcommand{\lket}[1]{{\ensuremath{\left\lvert#1\right\rangle}}}
\newcommand{\ket}[1]{\if@display\lket{#1}\else\sket{#1}\fi}

\newcommand{\sbra}[1]{{\ensuremath{\langle#1\rvert}}}
\newcommand{\lbra}[1]{{\ensuremath{\left\langle#1\right\rvert}}}
\newcommand{\bra}[1]{\if@display\lbra{#1}\else\sbra{#1}\fi}

\newcommand{\sbraket}[2]{{\ensuremath{\langle#1\rvert#2\rangle}}}
\newcommand{\lbraket}[2]{{\ensuremath{\left\langle#1\!\left\rvert\vphantom{#1}#2\right.\!\right\rangle}}}
\newcommand{\braket}[2]{\if@display\lbraket{#1}{#2}\else\sbraket{#1}{#2}\fi}

\newcommand{\sketbra}[2]{{\ensuremath{\lvert #1\rangle\!\langle #2\rvert}}}
\newcommand{\lketbra}[2]{{\ensuremath{\left\lvert #1\right\rangle\!\!\left\langle #2\right\rvert}}}
\newcommand{\ketbra}[2]{\if@display\lketbra{#1}{#2}\else\sketbra{#1}{#2}\fi}

\newcommand{\proj}[1]{\ketbra{#1}{#1}}

\newcommand{\hilbert}{\ensuremath{\mathcal{H}}}

\DeclareMathOperator{\tr}{tr}
\newcommand{\strace}[2][@]{\ensuremath{\tr\ifthenelse{\equal{#1}{@}}{}{_{#1}}(#2)}}
\newcommand{\ltrace}[2][@]{\ensuremath{\tr\ifthenelse{\equal{#1}{@}}{}{_{#1}}\left(#2\right)}}
\newcommand{\ktrace}[2][@]{\ensuremath{\tr\ifthenelse{\equal{#1}{@}}{}{_{#1}}\left[#2\right]}}
\newcommand{\trace}[2][@]{\if@display\ltrace[#1]{#2}\else\strace[#1]{#2}\fi}

\newcommand{\sspan}[1]{{\ensuremath{\operatorname{span}(#1)}}}
\newcommand{\lspan}[1]{{\ensuremath{\operatorname{span}\left(#1\right)\!}}}
\newcommand{\vspan}[1]{\if@display\lspan{#1}\else\sspan{#1}\fi}

\newcommand{\HmaxOp}{H_{\max}}
\newcommand{\HminOp}{H_{\min}}
\newcommand{\ImaxOp}{I_{\max}}
\newcommand{\DmaxOp}{D_{\max}} 
\newcommand{\pguessOp}{p_{\operatorname{guess}}}

\newcommand{\sHmin}[2][@]{\ensuremath{\HminOp(#2)\ifthenelse{\equal{#1}{@}}{}{_{#1}}}}
\newcommand{\lHmin}[2][@]{\ensuremath{\HminOp\left(#2\right)\ifthenelse{\equal{#1}{@}}{}{_{#1}}}}
\newcommand{\Hmin}[2][@]{\if@display\lHmin[#1]{#2}\else\sHmin[#1]{#2}\fi}

\newcommand{\sHminSmooth}[3][@]{\ensuremath{\HminOp^{#2}(#3)\ifthenelse{\equal{#1}{@}}{}{_{#1}}}}
\newcommand{\lHminSmooth}[3][@]{\ensuremath{\HminOp^{#2}\left(#3\right)\ifthenelse{\equal{#1}{@}}{}{_{#1}}}}
\newcommand{\HminSmooth}[3][@]{\if@display\lHminSmooth[#1]{#2}{#3}\else\sHminSmooth[#1]{#2}{#3}\fi}

\newcommand{\sHmax}[2][@]{\ensuremath{\HmaxOp(#2)\ifthenelse{\equal{#1}{@}}{}{_{#1}}}}
\newcommand{\lHmax}[2][@]{\ensuremath{\HmaxOp\left(#2\right)\ifthenelse{\equal{#1}{@}}{}{_{#1}}}}
\newcommand{\Hmax}[2][@]{\if@display\lHmax[#1]{#2}\else\sHmax[#1]{#2}\fi}

\newcommand{\sHmaxSmooth}[3][@]{\ensuremath{\HmaxOp^{#2}(#3)\ifthenelse{\equal{#1}{@}}{}{_{#1}}}}
\newcommand{\lHmaxSmooth}[3][@]{\ensuremath{\HmaxOp^{#2}\left(#3\right)\ifthenelse{\equal{#1}{@}}{}{_{#1}}}}
\newcommand{\HmaxSmooth}[3][@]{\if@display\lHmaxSmooth[#1]{#2}{#3}\else\sHmaxSmooth[#1]{#2}{#3}\fi}

\newcommand{\sImax}[2][@]{\ensuremath{\ImaxOp(#2)\ifthenelse{\equal{#1}{@}}{}{_{#1}}}}
\newcommand{\lImax}[2][@]{\ensuremath{\ImaxOp\left(#2\right)\ifthenelse{\equal{#1}{@}}{}{_{#1}}}}
\newcommand{\Imax}[2][@]{\if@display\lImax[#1]{#2}\else\sImax[#1]{#2}\fi}

\newcommand{\sImaxSmooth}[3][@]{\ensuremath{\ImaxOp^{#2}(#3)\ifthenelse{\equal{#1}{@}}{}{_{#1}}}}
\newcommand{\lImaxSmooth}[3][@]{\ensuremath{\ImaxOp^{#2}\left(#3\right)\ifthenelse{\equal{#1}{@}}{}{_{#1}}}}
\newcommand{\ImaxSmooth}[2][@]{\if@display\lImaxSmooth[#1]{#2}\else\sImaxSmooth[#1]{#2}\fi}

\newcommand{\sDmax}[2]{\ensuremath{\DmaxOp(#1\|#2)}}
\newcommand{\lDmax}[2]{\ensuremath{\DmaxOp\left(#1\middle\|#2\right)}}
\newcommand{\Dmax}[2]{\if@display\lDmax{#1}{#2}\else\sDmax{#1}{#2}\fi}

\newcommand{\spguess}[2][@]{\ensuremath{\pguessOp(#2)\ifthenelse{\equal{#1}{@}}{}{_{#1}}}}
\newcommand{\lpguess}[2][@]{\ensuremath{\pguessOp\left(#2\right)\ifthenelse{\equal{#1}{@}}{}{_{#1}}}}
\newcommand{\pguess}[2][@]{\if@display\lpguess[#1]{#2}\else\spguess[#1]{#2}\fi}

\newcommand{\no}[1]{\ensuremath{\mathcal{S}(\hilbert_{#1})}} 
\newcommand{\sno}[1]{\ensuremath{\mathcal{S}_\leq(\hilbert_{#1})}} 

\newcommand{\strnorm}[1]{\ensuremath{\|#1\|_{\tr}}}
\newcommand{\ltrnorm}[1]{\ensuremath{\left\|#1\right\|_{\tr}}}
\newcommand{\trnorm}[1]{\if@display\ltrnorm{#1}\else\strnorm{#1}\fi}

\DeclareMathOperator{\Fro}{Fro}
\newcommand{\sfronorm}[1]{\ensuremath{\|#1\|_{\Fro}}}
\newcommand{\lfronorm}[1]{\ensuremath{\left\|#1\right\|_{\Fro}}}
\newcommand{\fronorm}[1]{\if@display\lfronorm{#1}\else\sfronorm{#1}\fi}

\newcommand{\sinftynorm}[1]{\ensuremath{\|#1\|_{\infty}}}
\newcommand{\linftynorm}[1]{\ensuremath{\left\|#1\right\|_{\infty}}}
\newcommand{\inftynorm}[1]{\if@display\linftynorm{#1}\else\sinftynorm{#1}\fi}

\DeclareMathOperator{\polyOp}{poly} 
\newcommand{\poly}[1]{\if@display\lpoly{#1}\else\spoly{#1}\fi}
\newcommand{\spoly}[1]{\ensuremath{\polyOp (#1)}}
\newcommand{\lpoly}[1]{\ensuremath{\polyOp \left(#1\right)}}
\DeclareMathOperator{\suppOp}{supp} 
\newcommand{\supp}[1]{\if@display\lsupp{#1}\else\ssupp{#1}\fi}
\newcommand{\ssupp}[1]{\ensuremath{\suppOp (#1)}}
\newcommand{\lsupp}[1]{\ensuremath{\suppOp \left(#1\right)}}
\DeclareMathOperator{\GFOp}{GF} 
\newcommand{\GF}[1]{\if@display\lGF{#1}\else\sGF{#1}\fi}
\newcommand{\sGF}[1]{\ensuremath{\GFOp (#1)}}
\newcommand{\lGF}[1]{\ensuremath{\GFOp \left(#1\right)}}

\DeclareMathOperator*{\E}{\mathds{E}}  

\newcommand{\coloneqq}{:=}

\newcommand{\eps}{\varepsilon}

\newcommand{\cJ}{\mathcal{J}}
\newcommand{\cK}{\mathcal{K}}

\newcommand{\cP}{\mathcal{P}}

\newcommand{\cS}{\mathcal{S}}
\newcommand{\cT}{\mathcal{T}}

\newcommand{\cX}{\mathcal{X}}

\newcommand{\cZ}{\mathcal{Z}}

\theoremstyle{plain}
\newtheorem{thm}{Theorem}
\newtheorem{theorem}{Theorem}
\newtheorem{lem}[theorem]{Lemma}
\newtheorem{cor}[theorem]{Corollary}

\theoremstyle{definition}

\begin{document}
\title{Device-Independent Quantum Key Distribution with Local Bell Test}
\author{Charles Ci Wen \surname{Lim}}
\affiliation{Group of Applied Physics, University of Geneva, Switzerland.}
\author{Christopher \surname{Portmann}}
\affiliation{Group of Applied Physics, University of Geneva, Switzerland.}
\affiliation{Institute for Theoretical Physics, ETH Zurich, Switzerland.}
\author{Marco \surname{Tomamichel}}
\affiliation{Institute for Theoretical Physics, ETH Zurich, Switzerland.}
\affiliation{Centre for Quantum Technologies, National University of Singapore, Singapore.}
\author{Renato \surname{Renner}}
\affiliation{Institute for Theoretical Physics, ETH Zurich, Switzerland.}
\author{Nicolas \surname{Gisin}}
\affiliation{Group of Applied Physics, University of Geneva, Switzerland.}

\begin{abstract}
  Device-independent quantum key distribution (DIQKD) in its current
  design requires a violation of a Bell's inequality between two
  parties, Alice and Bob, who are connected by a quantum
  channel. However, in reality, quantum channels are lossy and current
  DIQKD protocols are thus vulnerable to attacks exploiting the
  detection-loophole of the Bell test. Here, we propose a novel
  approach to DIQKD that overcomes this limitation. In particular, we
  propose a protocol where the
  Bell test is performed entirely on two casually independent devices situated in Alice's laboratory. As a result,
  the detection-loophole caused by the losses in the channel is avoided.
\end{abstract}
\maketitle

\section{Introduction} 

The security of quantum key distribution (QKD)~\cite{bb84,ekert91}
relies on the fact that two honest parties, Alice and Bob, can devise
tests---utilizing laws of quantum physics---to detect any attack by an
eavesdropper, Eve, that would compromise the secrecy of the key strings they
generate~\cite{Gisin2002}. While the theoretical principles allowing
this are nowadays well understood, it turns out that realizing QKD
with practical devices is rather challenging. That is, the devices
must conform to very specific models, otherwise the implementation may
contain loopholes allowing side-channel
attacks~\cite{sidechannelrefs}.

In general, there are two broad approaches towards overcoming such
implementation flaws. The first is to include all possible
imperfections into the model used in the security analysis. This
approach, however, is quite cumbersome and it is unclear whether any
specific model includes all practically relevant imperfections.  In
the second approach, which is known as \emph{device-independent} QKD
(DIQKD)~\cite{pironio09,McKague2009,haenggi10,Haenggi10a, masanes10},
the security is based solely on the observation of non-local
statistical correlations, thus it is no longer necessary to provide
any model for the devices (though a few assumptions are still
required). In this respect, DIQKD appears to be the ultimate solution
to guarantee security against inadvertently flawed devices and
side-channel attacks.

DIQKD in its current design requires the two distant parties Alice and
Bob to perform a Bell test~\cite{Bell1964} (typically the
Clauser-Horne-Shimony-Holt (CHSH) test~\cite{Clauser1969}), which is
applied to pairs of entangled quantum systems shared between them.  In
practice, these quantum systems are typically realized by photons,
which are distributed via an optical fiber. Hence, due to losses
during the transmission, the individual measurements carried out on Alice's
and Bob's sites only succeed with bounded (and often small)
probability. In standard Bell experiments, one normally accounts for
these losses by introducing the \emph{fair-sampling} assumption, which
asserts that the set of runs of the experiment---in which both Alice's
and Bob's measurements succeeded---is representative for the set of all
runs.

In the context of DIQKD, however, the fair-sampling assumption is not
justified since Eve may have control over the set of detected
events. More concretely, she may use her control to emulate quantum
correlations based on a local deterministic model, i.e., she instructs
the detector to click only if the measurement setting (chosen by the party, e.g., Alice) is compatible
with the prearranged values. This problem is commonly known as the
detection-loophole~\cite{Pearle1970}. In fact, for state-of-the-art
DIQKD protocols, it has been shown in Ref.~\cite{Gerhardt2011} that the detection-loophole is
already unavoidable when using optical fibers of about 5km length.

One possible solution to this problem are heralded qubit
amplifiers~\cite{HQA}, which have been proposed recently. The basic
idea is to herald the arrival of an incoming quantum system without
destroying the quantum information it carries. This allows Alice and
Bob to choose their measurement settings only after receiving the
confirmation, which is crucial for guaranteeing security.
Unfortunately, realizing an efficient heralded qubit amplifier that is applicable for long distance DIQKD is extremely challenging; although there has been progress along this direction~\cite{Kocsis2012}.

In this work, we take a different approach to circumvent the detection-loophole. We propose a protocol that combines a self-testing scheme
for the Bennett and Brassard (BB84) states~\cite{bb84} with a protocol
topology inspired by the ``time-reversed'' BB84
protocol~\cite{Biham96, Inamori2002,Braunstein2012,Lo12}. Crucially,
the protocol only requires Bell tests carried out locally in Alice's
laboratory, so that the detection probabilities are not affected by the losses
in the channel connecting Alice and Bob.  We show that the protocol
provides device-independent security under the assumption that certain
devices are causally independent (see below for a more precise
specification of the assumptions).

In contrast to existing protocols for DIQKD, whose security is
inferred from the monogamy of non-local correlations, the security of
our protocol is proved using a recent generalization of the entropic
uncertainty relation that accounts for quantum side
information~\cite{berta10}.  This is the key ingredient that allows us
to circumvent the need to bound the non-locality between particle
pairs shared by Alice and Bob (non-locality over larger distances is hard to achieve, as explained
above). Instead, the uncertainty relation solely depends on the local
properties of the states sent by Alice, which in turn, can be
inferred from the local Bell test. 

Technically, our security proof uses a relation between the local CHSH
test and a variant of the entropic uncertainty relation for smooth
entropies~\cite{Tomamichel2010}. The analysis applies to the
(practically relevant) finite-size regime, where the secret key is
generated after a finite number of channel uses. The resulting bounds
on the achievable key size are comparable to the almost tight
finite-size result~\cite{Tomamichel2012a} for the BB84
protocol. Furthermore, in the (commonly studied) asymptotic limit
where the number of channel uses tends to infinity, and in the
limiting case where the CHSH inequality is maximally violated, the
performance of our protocol reaches the one of the BB84 protocol.

\section{Required Assumptions}

As mentioned above, our goal is to impose only limited and realistic
assumptions on the devices used by Alice and Bob. These are as
follows:

First, it is assumed that Alice and Bob's laboratories are perfectly
isolated, i.e., no information leaves a party's laboratory unless this
is foreseen by the protocol. Second, we assume that Alice and Bob can
locally carry out classical computations using trusted devices and
that they have trusted sources of randomness. Third, we assume that
Alice and Bob share an authenticated classical channel. Finally, we
require that the devices of Alice and Bob are causally independent,
that is, there is no correlation in their behavior between different
uses. This, for instance, is guaranteed if the devices have no
internal memory or if their memory can be reliably erased after each
use.

We remark that, in very recent work~\cite{RUV12,VV12,BCK12}, it has
been shown that this last assumption can be weakened further for
standard DIQKD protocols. More precisely, it is shown that the
assumption of causal independence can be dropped for the repeated uses
of a device within one execution of the protocol. However, the
assumption of causal independence still needs to be made when the same
devices are reused in a subsequent execution of the protocol, as
information about previously generated keys may leak
otherwise~\cite{BCK13}. 

\section{Protocol Topology}

In this section we describe the basic idea and the main structure of
the QKD scheme we propose. The actual protocol will then be detailed
in the next section.

Our proposal is motivated by the time-reversed BB84
protocol~\cite{Biham96, Inamori2002,Braunstein2012,Lo12}. This
protocol involves a third party, Charlie, whose task is to help Alice
and Bob distribute their key strings. Importantly, however, no trust in this
third party is required. While a deviation of Charlie from the
protocol may cause abortion of the key distribution protocol, it will
not compromise the secrecy of a successfully distributed key string. 
The time-reversed BB84 protocol consists of the followings steps:
first, Alice and Bob each generate a pair of qubits in the maximally
entangled state $\ket{\Phi^+}=\left(\ket{00}+\ket{11}\right)\sqrt{2}$
and send one of their qubits to Charlie. Subsequently, Charlie
performs a Bell-state-measurement (BSM) on the two received qubits and
broadcasts the outcome to Alice and Bob~\cite{Bennett1993}. The two
remaining qubits held by Alice and Bob are now in a Bell state. Alice
then applies appropriate bit and phase flips on her qubit to convert
this joint state to $\ket{\Phi^+}$. Finally, Alice and Bob measure
their qubits at random in one of the two BB84 bases. Note that Alice and Bob can alternatively measure the qubits they kept before Charlie performs the BSM, and Alice flips the outcome of her measurement if necessary once she has received the correction (i.e., the outcome of the BSM) from Charlie. 

The security of the time-reversed BB84 protocol, as described above,
depends on the correct preparation and measurement of the states by
Alice and Bob. In order to turn this protocol into a device
independent one, we add a CHSH test on Alice's site. Security is then
established by virtue of a relation between the violation of this CHSH
test and the incompatibility of the two possible measurements carried
out by Alice's device (which are supposed to be in the two BB84
bases). More precisely, we bound the overlap between the basis vectors
of the two measurements that Alice may choose. This is all that is
needed to apply the entropic uncertainty relation~\cite{Tomamichel2010} mentioned in the introduction, which
allows us to infer security without any further assumptions on Alice's
and Bob's devices. We note that our modification of the time-reversed
BB84 protocol is reminiscent of the idea of self-testing of devices
introduced by Mayers and Yao~\cite{MY04} (see Ref.~\cite{MYS12} for
the CHSH version). Our test has however a different purpose: its goal
is to certify the incompatibility of Alice's local measurements, while the test of Mayers and Yao certifies that Alice and Bob share a maximally entangled state. 

\begin{figure}[t]
\includegraphics[width=86mm]{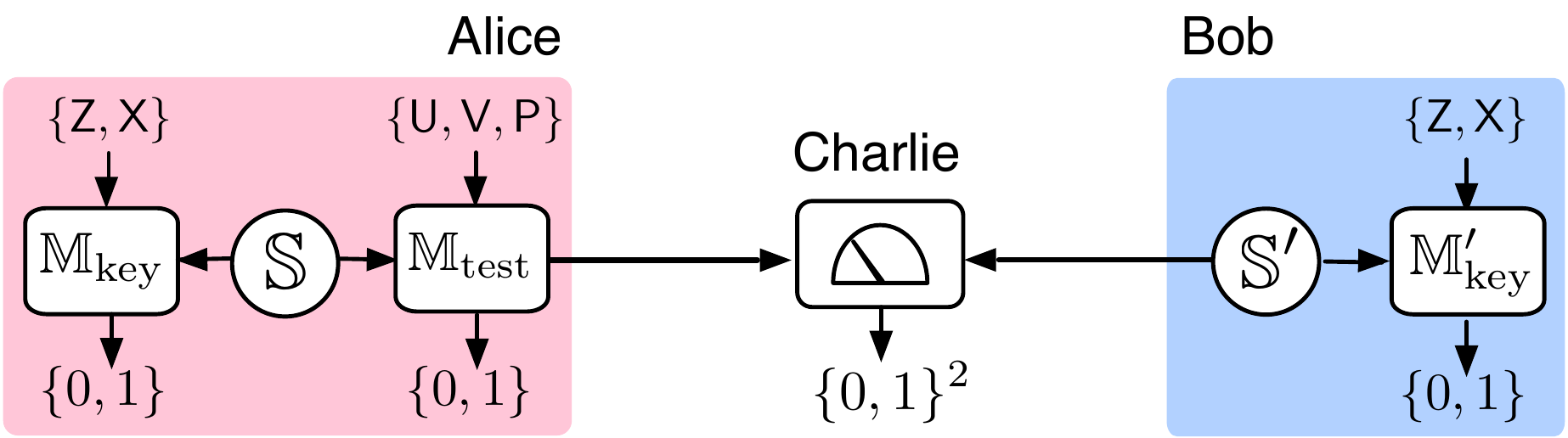}
  \caption{\textbf{Topology}. The protocol is inspired by the idea of the time-reversed BB84 protocol and involves an additional (untrusted) party, Charlie. Charlie is supposed to make an entangling measurement (ideally, a Bell-state-measurement) on quantum states sent by Alice and Bob. He outputs either a pass or fail to indicate whether the measurement was successful. If successful, he additionally outputs two bits to be used by Alice to make correcting bit-flip operations.}  
  \label{fig1}
\end{figure}

In order to realize the CHSH test, we use a setup with three different
devices on Alice's site: two measurement devices $\devicekey$,
$\devicetest$ and a source device $\devicesource$ (see
Fig.~\ref{fig1}).  The source device generates a pair of entangled
qubits and sends them to $\devicekey$ and $\devicetest$. The device
$\devicekey$ has two settings
$\{\mathsf{X},\mathsf{Z}\}$~\cite{footnote1} and produces a binary
output after one of the settings is chosen by Alice. The device
$\devicetest$ has three settings $\{\mathsf{U},
\mathsf{V},\mathsf{P}\}$. The first two produce a binary output (a measurement outcome), and the last one sends the qubit (from the device $\devicesource$) to the quantum channel which connects to Charlie. Alice has
therefore two modes of operation, of which one (corresponding to the
settings $\mathsf{U},\mathsf{V}$) is used to carry out the CHSH test
and one (corresponding to the setting $\mathsf{P}$) is chosen to communicate to Charlie. We refer to these operation modes as
$\Gamma_{\tn{CHSH}}$ and $\Gamma_{\tn{QKD}}$, respectively.

Bob has two devices: a measurement device $\devicekey'$ and a source
device $\devicesource'$. The latter generates entangled qubits and
sends one of them to the quantum channel and the other to 
$\devicekey'$. The device $\devicekey'$ has two settings
$\{\mathsf{X},\mathsf{Z}\}$ and produces a binary output after one of
the settings is chosen by Bob.


\section{Protocol Description} The protocol is parameterized by the secret key length $\ell$, the classical post-processing block size $m_x$, the error rate estimation sample size $m_z$, the local CHSH test sample size $m_j$, the tolerated CHSH value $\Stol$, the tolerated channel error rate $\Qtol$, the tolerated efficiency of Charlie's operation $\eta_{\tn{tol}}$, the error correction leakage $\leak$ and the required correctness $\epscorr$. 

In the following, the first three steps are repeated until the conditions in the sifting step are satisfied. 
\newline
  
\noindent\emph{1.~State preparation and distribution:} Alice selects
an operation mode $h_i\in\{ \Gamma_{\tn{CHSH}},\Gamma_{\tn{QKD}}\}$
where $\Gamma_{\tn{CHSH}}$ is selected with probability
$p_s=\eta_{\tn{tol}} m_j/\big(\eta_{\tn{tol}}
m_j+(\sqrt{m_x}+\sqrt{m_z})^2\big)$ and $\Gamma_{\tn{QKD}}$ is
selected with probability $1-p_s$~\cite{footnote2}. In the following,
we describe $\Gamma_{\tn{CHSH}}$ and $\Gamma_{\tn{QKD}}$ formally for
each of the runs, which we label with indices $i$.

$\Gamma_{\tn{CHSH}}$: Alice measures both halves of the bipartite state. More specifically, she chooses two bit values $u_i, v_i$ uniformly at random, where $u_i$ sets the measurement on the first half to $\mathsf{X}$ or $\mathsf{Z}$ and $v_i$ sets the measurement on the second half to $\mathsf{U}$ or $\mathsf{V}$. The outputs of each measurement are recorded in $s_i$ and $t_i$, respectively.

$\Gamma_{\tn{QKD}}$: Alice selects a measurement setting  $a_i\in\{ \mathsf{X},\mathsf{Z}\}$ with probabilities $p_x=1/(1+\sqrt{(m_z/m_x)})$ and $1-p_x$, respectively~\cite{footnote2}, measures one half of the bipartite state with it and stores the measurement output in $y_i$. The other half of the bipartite state is sent to Charlie.

Similarly, Bob selects a measurement setting  $b_i\in\{ \mathsf{X},\mathsf{Z}\}$ with probabilities $p_x$ and $1-p_x$, respectively, measures one half of the bipartite state with it and stores the measurement output in $y_i'$. The other half of the bipartite state is sent to Charlie. 
\newline

\noindent\emph{2.~Charlie's operation:} Charlie makes an entangling measurement on the quantum states sent by Alice and Bob, and if it is successful, he broadcasts  $f_i=\tn{pass}$, otherwise he broadcasts  $f_i=\tn{fail}$. Furthermore, if $f_i=\tn{pass}$, then Charlie communicates $g_i\in\{0,1\}^2$ to Alice and Bob. Finally, Alice uses $g_i$ to make correcting bit flip operations.
\newline

\noindent\emph{3.~Sifting:} Alice and Bob announce their choices $\{h_i\}_i, \{a_i\}_i, \{b_i\}_i$ over an authenticated classical channel and identify the following sets: key generation $\cX:= \{i:(h_i\allowbreak =\Gamma_{\tn{QKD}})\wedge (a_i=b_i=\mathsf{X})\wedge (f_i=\tn{pass})\}$, channel error rate estimation $\cZ:=\{i:(h_i=\Gamma_{\tn{QKD}})\wedge (a_i=b_i=\mathsf{Z})\wedge (f_i=\tn{pass})\}$, and Alice's local CHSH test set, $\cJ:=\{i:h_i=\Gamma_{\tn{CHSH}}\}$.

The protocol repeats steps (1)-(3) as long as $|\cX|<m_x$ or
$|\cZ|<m_z$ or $|\cJ|<m_j$, where $m_x,m_z,m_j \in \mathbb{N}$. We
refer to these as the sifting condition.  \newline

\noindent \emph{4.~Parameter estimation:} To compute the CHSH value from $\cJ$, Alice uses the following formula, $\Stes \allowbreak \coloneqq  \allowbreak 8 \allowbreak \sum_{i\in\cJ} f(u_i,v_i,s_i,t_i)/|\cJ|-4$,
where $f(u_i,v_i,s_i,t_i)=1$ if $s_i \oplus t_i= u_i \wedge v_i$, otherwise $f(u_i,v_i,s_i,t_i)=0$. Next, both Alice and Bob publicly announce the corresponding bit strings $\{y_i\}_{i\in\cZ}$, $\{y_i^\p\}_{i\in\cZ}$ and compute the error rate $\Qtes:=\sum_{i\in\cZ}y_i\oplus y_i^{\prime}/|\cZ|$. Finally, they compute the efficiency of Charlie's operation $\eta:=|\cX|/|\tilde{\cX}|$ where $\tilde{\cX}:=\{i:(h_i=\Gamma_{\tn{QKD}})\wedge (a_i=b_i=\mathsf{X})\}$.
If $\Stes < \Stol$ or $\Qtol<\Qtes$ or $\eta < \eta_{\tn{tol}}$, they abort the protocol. 
\newline

\noindent\emph{5.~One-way classical post-processing:} Alice and Bob
choose a random subset of size $m_x$ of $\cX$ for post-processing. An
error correction protocol that leaks at most $\leak$-bits of
information is applied, then an error verification protocol (e.g., this can be implemented with two-universal hashing) that leaks
$\lceil\log_2(1/\epscorr)\rceil$-bits of information is
applied. 
If the error verification fails, they abort the protocol. Finally, Alice and
Bob apply privacy amplification~\cite{Bennett1995} with two-universal
hashing to their bit strings to extract a secret
key of length $\ell$~\cite{RennerThesis}.

\section{Security definition} 
Let us briefly recall the criteria for a generic QKD protocol to be secure. A QKD protocol either aborts or provides Alice and Bob with a pair of key strings, $S_A$ and $S_B$, respectively. If we denote by $E$ the information that the eavesdropper (Eve) gathers during the protocol execution, then the joint state of $S_A$ and $E$ can be described by a classical-quantum state, $\rho_{S_AE}= \allowbreak \sum_s\proj{s}\otimes \rho_E^s$ where $\{\rho_E^s\}_s$ are quantum systems (conditioned on $S_A$ taking values $s$) held by Eve.  The QKD protocol is called $\epscorr$-correct if $\Pr[S_A \not = S_B] \leq \epscorr$, and $\esec$-secret if $(1-p_\tn{abort} )\frac{1}{2}\|\rho_{S_AE}-U_{S_A} \otimes \rho_E \|_1 \leq \esec$ where $p_\tn{abort}$ is the probability that the protocol aborts and $U_{S_A}$ is the uniform mixture of all possible values of the key string $S_A$. Accordingly, we say that the QKD protocol is $(\epscorr+\esec)$-secure if it is both $\epscorr$-correct and $\esec$-secret~\cite{Tomamichel2012a,MullerQuadeRenner,RennerThesis}. Note that this security definition guarantees that the QKD protocol is universally composable~\cite{RennerThesis, MullerQuadeRenner}. That is, the pair of key strings can be safely used in any application (e.g., for encrypting messages) that requires a perfectly secure key (see~\cite{RennerThesis} for more details).

\section{Security analysis}\label{section:SA} 
In the following, we present the main result and a sketch of its proof. For more details about the proof, we refer to the Appendix.

The correctness of the protocol is guaranteed by the error verification protocol which is parameterized by the required correctness $\epscorr$. \newline

\textbf{Main Result.}~\emph{The protocol with parameters $(\ell, m_x,\allowbreak m_z, \allowbreak m_j, \allowbreak \Stol, \allowbreak \Qtol, \eta_{\tn{tol}}, \allowbreak \leak, \allowbreak \epscorr)$ is $\esec$-secret if  \begin{multline} \label{eqn1}
 \ell \leq  m_x \bigg( 1 - \log_2
    \bigg(1+\frac{\hat{S}_{\tn{tol}}}{4\eta_{\tn{tol}}} \sqrt{8 -
      \hat{S}_{\tn{tol}}^{2}} + \frac{\zeta}{\eta_{\tn{tol}}}\bigg) - \tn{h}(\hat{Q}_{\tn{tol}})
    \bigg)  \\ \hspace{2cm} -\leak  -\log_2\frac{1} {\epscorr\eps^4},\end{multline}
  for $\eps \allowbreak =\esec/9$ and $2 \leq \hat{S}_{\tn{tol}} \leq 2\sqrt{2}$, where $\tn{h}$ denotes the binary entropy function, $\hat{S}_{\tn{tol}} :=\Stol- \xi$ and $ \hat{Q}_{\tn{tol}}:= Q_{\tn{tol}}  + \mu$, with the statistical deviations given by
 \begin{align*}
    \xi  &\coloneqq \sqrt{\frac{32}{m_j}\ln \frac{1}{\eps } }, \\
    \zeta& \coloneqq \sqrt{\frac{2(m_x+m_j\eta)(m_j+1)}{m_xm_j^2}\ln \frac{1}{\eps}},\\  
    \mu &  \coloneqq \sqrt{\frac{(m_x+m_z)(m_z+1)}{m_x m_z^2} \ln
      \frac{1}{\eps}}. \end{align*}
} 

\emph{Proof sketch.} Conditioned on passing all the tests in the
parameter estimation step, let $X_A$ be the random variable of length
$m_x$ that Alice gets from ${\cX}$ and let $E'$ denote Eve's
information about $X_A$ at the end of the error correction and error
verification protocols.

We use the following result from~\cite{RennerThesis}.  By using
privacy amplification with two-universal hashing, a $\Delta$-secret
key of length $\ell$ can be generated from $X_A$ with 
\[
\Delta \leq 6\eps+2^{-\frac{1}{2}\big(\HminSmooth{3\eps}{X_A|E'}-\ell\big)-1}
\]
for any $\eps > 0$. Here $\HminSmooth{3\eps}{X_A|E'}$ denotes the
smooth min-entropy~\cite{RennerThesis}. It therefore suffices to bound
this entropy in terms of the tolerated values ($\Stol$, $\Qtol$ and
$\eta_{\tn{tol}}$).

First, using chain rules for smooth entropies~\cite{RennerThesis}, we get 
$
\HminSmooth{3\eps}{X_A|E'} \geq
\HminSmooth{3\eps}{X_A|E} - \tn{leak}_{\tn{EC}}-\log_2(2/\epscorr)$, where $E$ denotes Eve's information after the parameter estimation step. Then, from the generalised entropic uncertainty relation~\cite{esther11}, we further get 
\[\HminSmooth{3\eps}{X_A|E} \geq \log_2 \frac{1}{c^*} -
\HmaxSmooth{\eps}{Z_A|Z_B} - \log_2 \frac{2}{\eps^2},\] 
where $c^*$ is the effective overlap of Alice's measurements (a function of the measurements corresponding to settings $\mathsf{Z},\mathsf{X}$ and the marginal state). Here, $Z_A$ can be seen as the bit string Alice would have obtained if she had chosen setting $\mathsf{Z}$ instead. Likewise, $Z_B$ represents the bit string obtained by Bob with setting $\mathsf{Z}$. From Ref.~\cite{Tomamichel2012a}, the smooth max-entropy of the alternative measurement is bounded by the error rate sampled on the set $\cZ$ of size $m_z$, 
$\HmaxSmooth{\eps}{Z_A|Z_B} \leq m_x \tn{h}(Q_{\tn{tol}} + \mu)$,
where $\mu$ is the statistical deviation due to random sampling theory, i.e., with high probability, the error rate between $Z_A$ and $Z_B$ is smaller than $Q_{\tn{tol}} + \mu$. 

It remains to bound the effective overlap $c^*$ with $\Stol$ and
$\eta_{\tn{tol}}$. First, we note that $\tilde{\cX}$ is independent of
Charlie's outputs and $\cX\subseteq \tilde{\cX}$ with equality only if
Charlie always outputs a pass. Furthermore, $\cX$ is not necessarily a
random subset of $\tilde{\cX}$ as a malicious Charlie can control the
content of $\cX$ (this is discussed later). Assuming the worst case
scenario, it can be shown that $ c^*\leq
1/2+\left(\tilde{c}^*-1/2\right)/\eta,$ where
$\eta=|\cX|/|\tilde{\cX}|$ is the efficiency of Charlie's operation
and $\tilde{c}^*$ is the effective overlap of $\tilde{\cX}$. Next, by
establishing a relation between the effective overlap and the local
CHSH test~\cite{esther11} (for completeness, we provide a more concise
proof in Lemma~\ref{lem:chsh_c} in the Appendix) and
using random sampling theory, we further obtain
\[\tilde{c}^*\leq	\frac{1}{2}\bigg(1+\frac{ (S_{\tn{tol}}-\xi)}{4}\sqrt{8- \big(S_{\tn{tol}}-\xi\big)^2}+\zeta\bigg).\] 
Here $\xi$ quantifies the statistical deviation between the expected CHSH value and the observed CHSH value, and $\zeta$ quantifies the statistical deviation between the effective overlaps of $\tilde{\cX}$ and $\cJ$, respectively.

\begin{figure}[t!]
  \centering  
  \includegraphics[width=87mm]{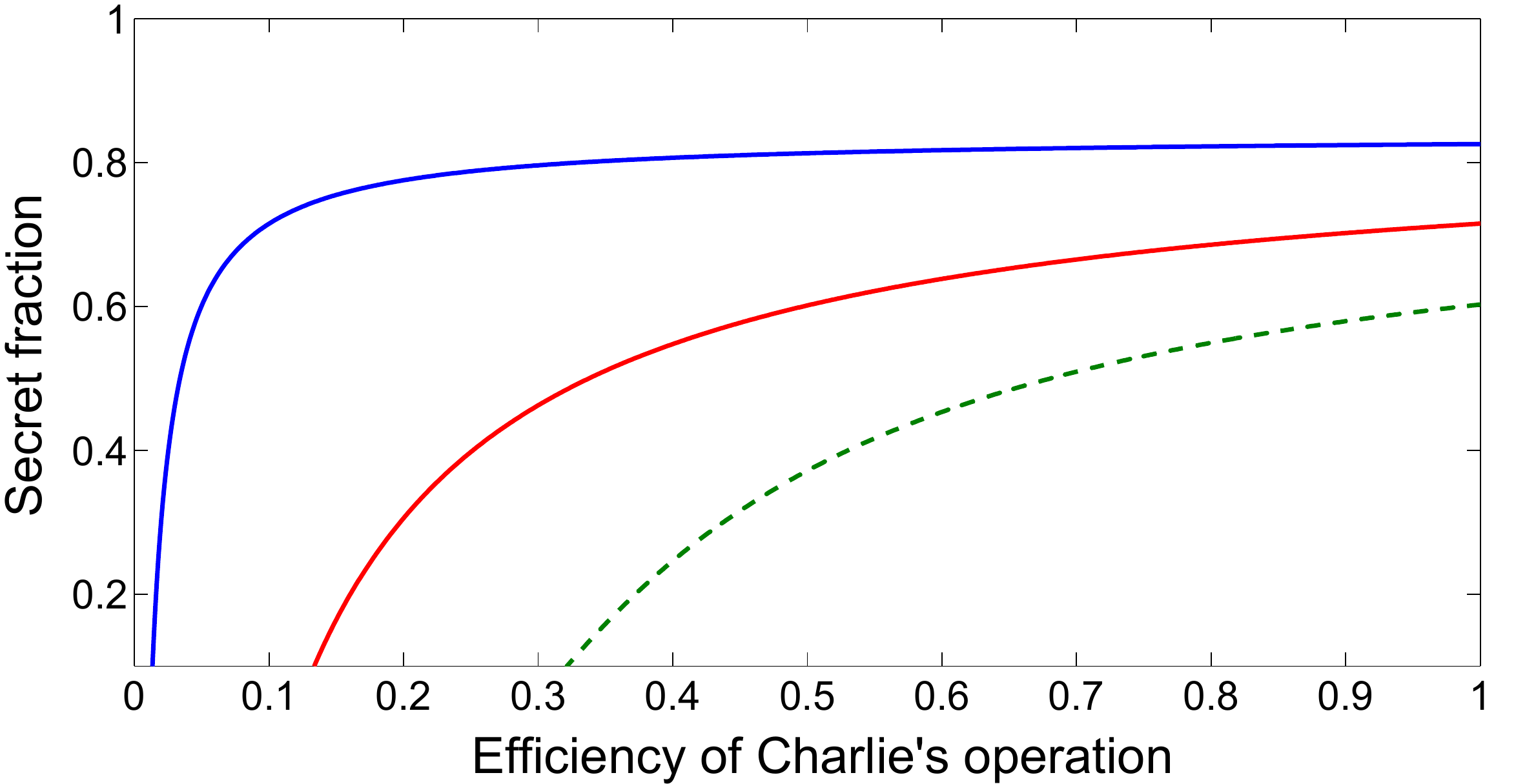}
  \caption{\textbf{Secret fraction $\ell/m_x$ as a function of the tolerated efficiency of Charlie's operation $\eta_{\tn{tol}}$ (including channel losses).} We consider a depolarising channel with a fixed error rate $\Qtol=1\%$ and $\Stol=V2\sqrt{2}$. The solid curves (asymptotic rates, Eq.~(\ref{eqn2})) are obtained with $V=0.999$ and $V=0.99$ from left to right. The right dashed curve (finite-key analysis, Eq.~(\ref{eqn1})) is obtained by choosing $\Stol=2\sqrt{2}$, $\leak=m_x1.1\tn{h}(\Qtol+\mu)$, $\esec=10^{-8}$ and $\epscorr=10^{-12}$, where the classical post-processing block size $m_x$ is of the order of $10^8$ bits. }\label{Fig2}
\end{figure}

Putting everything together, we obtain the secret key length as stated
by Eq.~(\ref{eqn1}).
\newline

\emph{Asymptotic limit.} In the following, we consider the secret
fraction defined as $f_{\tn{secr}}:=\ell/m_x$~\cite{Gisin2002}. In the
asymptotic limit $N \rightarrow \infty$ and using $\leak \rightarrow
\tn{h}(\Qtol)$ (corresponding to the Shannon limit), it is easy to
verify that the secret fraction reaches
\be \label{eqn2}f_{\tn{secr}}=1 - \log_2 \left (1+
  \frac{\Stol}{4\eta_{\tn{tol}}} \sqrt{8 - \Stol^{2}}
\right)-2\tn{h}(\Qtol).\ee The expression reveals the roles of the
modes of operation $\Gamma_{\tn{CHSH}}$ and $\Gamma_{\tn{QKD}}$. The
first provides a bound on the quality of the devices (which is taken
into account by the $\log_2$ term) and the latter, apart from
generating the actual key, 
is a measure for the quality of the quantum channel.

\section{Discussion} 

We have proposed a DIQKD protocol which provides security even if the
losses of the channel connecting Alice and Bob would not allow for a
detection-loophole free Bell test. Nevertheless, the security of the
protocol still depends on the losses and the protocol therefore
needs to perform a check to ensure that Charlie does not output a fail too
often. This dependence from the failure probability arises from the
fact that a malicious Charlie may choose to output a pass only when
Alice and Bob's devices behave badly. 
Therefore, the CHSH value calculated from Alice's CHSH sample is not a reliable estimate for the overlap of the sample used to generate the key string. However, with the CHSH test, Alice
can estimate how often her devices behave badly and thus determine the minimum tolerated efficiency (or the maximum tolerated failure probability) of Charlie. This is illustrated in Fig.~\ref{Fig2} where
large values of $\Stol$ are required to tolerate small values of
$\eta_{\tn{tol}}$.

Taking the asymptotic limit and the maximal CHSH value, we see that
the secret fraction is independent of $\eta_{\tn{tol}}$, which is not
so surprising since the maximal CHSH value implies that the
devices of Alice are behaving ideally all the time. Remarkably, we recover the
asymptotic secret fraction for the BB84 protocol~\cite{SP00}. 

From a practical point of view, the possibility to consider very small
values of $\eta_{\tn{tol}}$ is certainly appealing, since it suggests
that the distance between Alice and Bob can be made very large. A
quick calculation using the best experimental values~\cite{ExpMDIQKD} (i.e.,
$\eta_{\tn{tol}}\approx t/2$ and
$\Stol \approx 2.81$ where $t$ is the channel transmission) shows that
the secret fraction is positive for $t > 0.45$. This translates to
about a $17$km optical fiber between Alice and Bob. Accordingly, to
achieve larger distances, we would need a local CHSH test that
generates violations larger than those achieved by current
experiments.

\section{Conclusion}

In summary, we provide an alternative approach towards DIQKD, where
the Bell test is not carried out between Alice and Bob but rather in
Alice's laboratory. On a conceptual level, our approach departs from
the general belief that the observation of a Bell violation between
Alice and Bob is necessary for DIQKD. On the practical side, it offers
the possibility to replace the extremely challenging task of
implementing a long distance detection-loophole free Bell test with a
less challenging task, i.e., implementing a local detection-loophole free Bell test. In fact, recently, there has been very
encouraging progress towards the implementation of a local detection-loophole free CHSH test~\cite{Giustina2013}. In view of that, we
believe an experimental demonstration of DIQKD with local Bell tests
is plausible in the near future.  \newline

\textbf{Acknowledgments.} We thank Jean-Daniel Bancal, Marcos Curty,
Esther H\"{a}nggi, Stefano Pironio, Nicolas Sangouard, Valerio Scarani
and Hugo Zbinden for helpful discussions. We acknowledge support from
the National Centre of Competence in Research QSIT, the Swiss NanoTera
project QCRYPT, the Swiss National Science Foundation SNSF (grant
No. 200020-135048), the CHIST-ERA project DIQIP, the FP7 Marie-Curie
IAAP QCERT project and the European Research Council (grant No. 258932
and No. 227656).

\appendix
\setcounter{equation}{0}
\setcounter{thm}{0}
\section*{Appendix: D\lowercase{etails of Security analysis}} 
We present the proof for the main result given in the main text. First, we discuss about the assumptions and then introduce the necessary technical lemmas. Second, we establish a relation between the local CHSH test and a generalized version of smooth entropic uncertainty relation (Lemma \ref{lem:chsh_c}). Third, we provide the required statistical statements for estimating certain quantities of the bit strings of Alice and Bob. Finally, we state our main result (Theorem \ref{securitythm}) which is slightly more general than the main result presented above. 

\subsection{Notations}
We assume that all Hilbert spaces denoted by $\hilbert$, are finite-dimensional. For composite systems, we define the tensor product of $\hilbert_A$ and $\hilbert_B$ as $\hilbert_{AB}:=\hilbert_A \otimes \hilbert_B$. We denote $\cP(\hilbert)$ as the set of positive semi-definite operators on $\hilbert$ and $\cS(\hilbert)$ as the set of normalised states on $\hilbert$, i.e., $\cS(\hilbert)=\{\rho\in\cP(\hilbert):\trace{\rho}=1\}$. Furthermore, for a composite state $\rho_{AB}\in \cS(\hilbert_{AB})$, the reduced states of system A and system B are given by $\rho_A=\trace[B]{\rho_{AB}}$ and $\rho_B=\trace[A]{\rho_{AB}}$, respectively. A positive operator valued measure (POVM) is denoted by $\mathds{M}:=\{M_x\}_x$ where $\sum_xM_x=\mathds{1}$. For any POVM, we may view it as a projective measurement by introducing an ancillary system, thus for any POVM with binary outcomes, we may write it as an observable $O=\sum_{x\in\{0,1\}} (-1)^x M_x$, such that $\sum_{x\in\{0,1\}} M_x=\mathds{1}$. We also use $\bar{x}:=(x_1,x_2,\ldots,x_n)$ to represent the concatenations of elements and $[n]$ to denote $\{1,2,\ldots,n\}$. The binary entropy function is denoted by $\tn{h}(x):=-x\log_2x-(1-x)\log_2(1-x)$.

\subsection{Basic assumptions on Alice's and Bob's abilities}\label{Ass}
Prior to stating the security proof, it is instructive to elucidate the basic assumptions necessary for the security proof. In particular, the assumptions are detailed in the following:

\begin{enumerate}
 \item \textbf{Trusted local sources of randomness.} Alice (also Bob) has access to a trusted source that produces a random and secure bit value upon each use. Furthermore, we assume the source is unlimited, that is, Alice can use it as much as she wants, however the protocol only requires an amount of randomness linear in the number of quantum states generated. \label{A1}

 \item \textbf{An authenticated but otherwise insecure classical channel.}\label{A2} Generally, this assumption is satisfied if Alice and Bob share an initial short secret key~\cite{WC81, Stinson94}. Note that the security analysis of such authentication schemes was recently extended to the universally composable framework~\cite{RennerThesis,MullerQuadeRenner} in Ref~\cite{Portmann2012}, which allows one to compose the error of the authentication scheme with the errors of the protocol, giving an overall error on the security.

  \item \textbf{No information leaves the laboratories unless the protocol allows it.}\label{A3} This assumption is paramount to any cryptographic protocol. It states that information generated by the legitimate users is appropriately controlled. More concretely, we assume the followings
\begin{enumerate}
  \item \emph{Communication lines.---} The only two communication lines leaving the laboratory are the classical and the quantum channel. Furthermore, the classical channel is controlled, i.e., only the information required by the protocol is sent.
  \item \emph{Communication between devices.---} There should be no unauthorized communication between any devices in the laboratory, in particular from the measurement devices to the source device. 
\end{enumerate}

\item \textbf{Trusted classical operations.} Classical operations like authentication, error correction, error verification, privacy amplification, etc must be trusted, i.e., we know that the operations have ideal functionality and are independent of the adversary.\label{A4}

\item  \textbf{Measurement and source devices are causally independent.}  \label{A5} This means each use of the device is independent of the previous uses. For example, for $N$ uses of a source device and a measurement that produces a bit string $\bar{x}:=(x_1,x_2,\ldots,x_n)$, we have \[\rho^N=\bigotimes_{i=1}^N\rho^i,\quad M_{\bar{x}} = \bigotimes_i M^i_{x_i}\] where $ M_{\bar{x}}$ is the POVM element corresponding to the outcome $\bar{x}$. 
\end{enumerate}

\subsection{Technical lemmas}

\begin{lem}[Jordan's lemma~\cite{Masanes2006, pironio09}]\label{lem:jordan} Let $O$ and $O^{\prime}$ be observables with eigenvalues $\pm 1$ on Hilbert space $\hilbert$. Then there exists a partition of the Hilbert space, $\hilbert=\bigoplus_i\hilbert_i$, such that
\[
O=\bigoplus_{i}O_i \quad \tn{and} \quad  O^{\prime}=\bigoplus_iO^{\prime}_i 
\]
where $\hilbert_i$ satisfies $\tn{dim}(\hilbert_i) \leq 2$ for all $i$.
\end{lem}

\begin{lem}[Chernoff-Hoeffding~\cite{hoeffding63}]
\label{lem:chernoff}
Let $X \coloneqq \frac{1}{n} \sum_i X_i$ be the average of $n$
\emph{independent} random variables $X_1,X_2,\dotsc,X_n$ with values
in $[0,1]$, and let $\mu \coloneqq \E[X] = \frac{1}{n} \sum_i \E[X_i]$
denote the expected value of $X$. Then, for any $\delta > 0$,
\[\Pr \left[X - \mu \geq \delta \right] \leq \exp(-2\delta^2n).\]
\end{lem}

\begin{lem}[Serfling~\cite{serfling74}]
\label{lem:serfling}
Let $\{x_1,\dotsc,x_n\}$ be a list of (not necessarily distinct)
values in $[a,b]$ with average $\mu \coloneqq \frac{1}{n}\sum_i x_i$. Let the
random variables $X_1,X_2,\dotsc,X_k$ be obtained by sampling $k$
random entries from this list without replacement. Then, for any $\delta >
0$, the random variable $X \coloneqq \frac{1}{k} \sum_i X_i$ satisfies
\[\Pr \left[X - \mu \geq \delta \right] \leq  \exp
\left(\frac{-2\delta^2kn}{(n-k+1)(b-a)} \right).\]
\end{lem}

\begin{cor} \label{Cor:serfling}Let $\cX:=\{x_1,\dotsc,x_n\}$ be a list of (not necessarily distinct) values in $[0,1]$ with the average $\mu_{\cX}:=\frac{1}{n}\sum_{i=1} x_i$. Let $\cT$ of size $k$ be a random subset of $\cT$ with the average $\mu_{\cT}:=\frac{1}{t}\sum_{i\in \cT} x_i$. Then for any $\eps > 0$, the set $\cK=\cX \setminus \cT$ with average $\mu_{\cK}=\frac{1}{n-t}\sum_{i\in\cK} x_i$ satisfies
\[\Pr\left[\mu_{\cK}-\mu_{\cT}\geq \sqrt{\frac{n(t+1)}{2(n-t)t^2}\ln\frac{1}{\eps}}\right] \leq \eps
\]
\end{cor}
\begin{proof} Since $\cT$ is a random sample of $\cX$, from Lemma \ref{lem:serfling}, we have
\[\Pr \left[\mu_{\cK} - \mu_{\cX} \geq \delta \right] \leq  \exp
\left(\frac{-2\delta^2(n-t)n}{(t+1)} \right)=\eps.\]
Using $\mu_{\cX}=\frac{t}{n}\mu_{\cT}+\frac{n-t}{n}\mu_{\cK}$ we finish the proof.
\end{proof}

The main ingredient is a fine-grained entropic uncertainty relation (see~\cite[Corollary~7.3]{Tomamichel2012thesis} and~\cite{esther11}).
\begin{lem}
\label{thm:ucr}
Let $\eps >0, \bar{\eps} \geq 0$ and $\rho \in \sno{ABC}$. Moreover let
$\mathds{M} = \{M_x\}$, $\mathds{N} = \{N_z\}$ be POVMs on
$\hilbert_A$, and $\mathds{K} = \{P_k\}$ a projective measurement on
$\hilbert_A$ that commutes with both $\mathds{M}$ and $\mathds{N}$. Then the post-measurement states
$\rho_{XB}= \sum_{x} \proj{x} \otimes \trace[AC]{\sqrt{M_x}
  \rho_{ABC}\sqrt{M_x}}$, $
  \rho_{ZC} = \sum_{z} \proj{z} \otimes\trace[AB]{\sqrt{N_z}
  \rho_{ABC}\sqrt{N_z}}
$ satisfy
\begin{multline}
  \label{eq:ucr}
  \HminSmooth[\rho]{2\eps+\bar{\eps}}{X|B} +
  \HmaxSmooth[\rho]{\eps}{Z|C}  \\ \geq \log_2
  \frac{1}{c^*(\rho_A,\mathds{M},\mathds{N})} - \log_2 \frac{2}{\bar{\eps}^2},
\end{multline}
where the effective overlap is defined as
\begin{multline}
\label{eq:effective}
c^*(\rho_A,\mathds{M},\mathds{N})  \\  \coloneqq  \min_{\mathds{K}}  \left\{\sum_k
\trace{P_k \rho} \max_x \inftynorm{P_k\sum_z
  N_{z}M_{x}N_{z}} \right\}
\end{multline}
\end{lem} Note that (\ref{eq:ucr}) is a statement about the entropies of the post-measurement states $\rho_{XB}$ and $\rho_{ZC}$, thus it also holds for any measurements that lead to the same post-measurement states. Accordingly, one may also consider the projective purifications $\mathds{M}'$ and $\mathds{N}'$ of $\mathds{M}$ and $\mathds{N}$, applied to $\rho_A \otimes
\proj{\phi}$, where $\ket{\phi}$ is a pure state of an ancillary system. Since both measurement setups $\{\rho, \mathds{M},\mathds{N}\}$ and $\{\rho_A \otimes
\proj{\phi},\mathds{M}', \mathds{N}'\}$ give the same post-measurement states, the R.H.S of (\ref{eq:ucr}) holds for both $c^*(\rho_A,\mathds{M},\mathds{N})$ and $c^*(\rho_A \otimes \proj{\phi},\mathds{M}',\mathds{N}')$. We can thus restrict our considerations to projective measurements.

In the protocol considered, Alice performs independent binary
measurements --- $\mathds{M}_i = \{M^i_{x}\}_{x \in \{0,1\}}$ and
$\mathds{N}_i = \{N^i_{z}\}_{z \in \{0,1\}}$ --- on each subsystem
$i$. We can reduce \eqref{eq:effective} to operations on each
subsystem, if we choose $\mathds{K} = \{P_{\bar{k}}\}$ to also be in
product form, i.e., $P_{\bar{k}} = \bigotimes_i P^i_{k_i}$, where
$\bar{k}$ is a string of (not necessarily binary) letters $k_i \in
\cK$. Then plugging this, $M_{\bar{x}} = \bigotimes_i M^i_{x_i}$ and
$N_{\bar{z}} = \bigotimes_i N^i_{z_i}$ in the norm from
\eqref{eq:effective}, we get
\begin{multline}
\inftynorm{P_{\bar{k}}\sum_{\bar{z}}
  N_{\bar{z}}M_{\bar{x}}N_{\bar{z}}} \\ =  \inftynorm{\sum_{z_1,z_2,\dots} \bigotimes_i P_{k_i}
N^i_{z_i}M^i_{x_i}N^i_{z_i} }
\\  =  \prod_i \inftynorm{P_{k_i}\sum_{z_i}
N^i_{z_i}M^i_{x_i}N^i_{z_i} }.
\end{multline}

Putting this in \eqref{eq:effective} with $\rho=\bigotimes_i\rho^i$, $p^i_k:=\tr(P_k^i \rho^i)$, and dropping the subscript $i$ when possible, we obtain,
\begin{align}
& c^*(\rho_A,\mathds{M},\mathds{N}) \notag  \leq \sum_{k_1,k_2,\dots}
\prod_i p^i_{k_i} \max_{x} \inftynorm{P_{k_i}\sum_{z}
N^i_{z}M^i_{x}N^i_{z}} \notag \\
  & \qquad = \prod_i \sum_k p^i_k \max_{x} \inftynorm{P^i_k\sum_{z}
N^i_{z}M^i_{x}N^i_{z}}=: \prod_i c^{*,i}. \label{eq:effective.subsystem}
\end{align}

In the following we will refer
to \begin{equation} \label{eq:effective.decomposition} c^{i}_k
  \coloneqq \max_{x} \inftynorm{P^i_k\sum_{z}
    N^i_{z}M^i_{x}N^i_{z}} \end{equation}
as the overlap of the measurements $\{M^i_{x}\}_{x}$ and
$\{N^i_{z}\}_{z}$.

\subsection{An upper bound on the effective overlap with the CHSH value}

In this section, we first introduce the notion of CHSH operator~\cite{BelloperatorRefs} and then prove the relation between the CHSH test and the effective overlap (\ref{eq:effective.decomposition}).

In the CHSH test, two space-like separated systems share a bipartite state $\rho$ and each system has two measurements. More specifically, system A has POVMs $\{M_0^0,M_1^0\}$ and $\{M_0^1,M_1^1\}$ and system T has POVMs $\{T_0^0,T_1^0\}$ and $\{T_0^1,T_1^1\}$. Since for any POVM there is a (unitary and) projective measurement on a larger Hilbert space that has the same statistics, we can restrict our considerations to projective measurements. Then, we may write the POVMs as observables with $\pm1$ outcomes, i.e., at the site of the first system, the two observables are $O_A^0:=\sum_{s=0}^1(-1)^sM_s^0$ and $O_A^1:=\sum_{s=0}^1(-1)^sM_s^1$. Furthermore, the measurements are chosen uniformly at random. As such, the CHSH value is given by $S(\rho,\beta):=\Tr(\rho\beta)$ where the CHSH operator is defined as
\begin{align}\label{CHSH:op}
 \beta(O_A^0,O_A^1,O_T^0,O_T^1) &:=\sum_{u,v}(-1)^{u \wedge v}O_A^u \otimes O_T^v 
\end{align}
where $u,v$ and $s,t$ are the inputs and outputs, respectively. The maximization of $S(\rho,\beta)$ over the set of density operators for a fixed $\beta$ is defined by $S_{\tn{max}}(\beta)$. 
Moreover, the CHSH operator can be decomposed into a direct sum of two-qubits subspaces via Lemma \ref{lem:jordan}. Mathematically, we may write $O_A^0=\sum_kP_kO_A^0P_k$ and $O_A^1=\sum_kP_kO_A^1P_k$ where $\{P_k\}_k$ is a set of projectors such that $\tn{dim}(P_k) = 2 ~\forall~k$. Note that in Lemma \ref{lem:jordan}, one may select a partition of the Hilbert space such that each block partition has dimension two. This allows one to decompose the general CHSH operator into direct sums of qubits CHSH operators. Likewise, for the measurements of Bob, $O_B^0=\sum_rQ_rO_T^0Q_r$ and $O_B^1=\sum_rQ_rO_T^1Q_r$. For all $k$, $P_kO_A^0P_k$ and $P_kO_A^1P_k$ can be written in terms of Pauli operators,
 \begin{align}
 P_kO_A^0P_k=\vec{m}_k\cdot \Gamma_k \quad \tn{and}\quad P_kO_T^1P_k=\vec{n}_k\cdot \Gamma_k, \label{CHSH:block}
 \end{align}
where $\vec{m}_k$ and $\vec{n}_k$ are unit vectors in $\mathbb{R}_k^3$ and $\Gamma_k$ is the Pauli vector. Combining \eqref{CHSH:op} and \eqref{CHSH:block} yields
\be 
\beta=\bigoplus_{k,r}\beta_{k,r} \quad \tn{where}\quad \beta_{k,r} \in \mathds{C}^2_k\otimes\mathds{C}^2_r
\label{CHSH:twoqubits}
\ee
and it can be verified that 
\begin{align}\label{decom:CHSH}
S(\rho,\beta)&=\sum_{k,r}\lambda_{k,r}S_{k,r}
\end{align}
where 
\begin{align} \label{CHSH:lambda_kr}
&\lambda_{k,r}:=\Tr(P_k\otimes Q_r \rho) \\
\label{CHSH:w_kr}
&S_{k,r}:=\Tr(\rho_{k,r}\beta_{k,r}) 
\end{align}
Whenever the context is clear, we write $S=S(\rho,\beta)$ and $S_{\tn{max}}=S_{\tn{max}}(\beta)$. 

In the following analysis, we consider only one subsystem, the superscript $i$ is omitted, i.e., we use $c^*=\sum_kp_kc_k$ instead.
\newline
\begin{lem}\label{lem:chsh_c}
Let $\{O_A^x\}_{x \in \{0,1\}} $ and  $\{O_T^y\}_{y \in \{0,1\}} $ be observables with eigenvalues $\pm 1$ on $\hilbert_A$ and $\hilbert_T$ respectively and let $\beta=\sum_{x,y}(-1)^{x \wedge y}O_A^x \otimes O_T^y$ be the CHSH operator. Then for any $\rho\in\no{AT}$, the effective overlap $c^*$ is related to the CHSH value $S=\Tr(\rho \beta)$ by
\be
c^* \leq \frac{1}{2} + \frac{S}{8}\sqrt{8-S^2}
\ee
\end {lem}
\begin{proof}
Using \eqref{CHSH:block}, let the relative angle between $\vec{m}_k$ and $\vec{n}_k$ be $\theta_k \in [0,\pi/2]$ for all $k$, i.e., $\vec{m}_k \cdot \vec{n}_k=\cos(\theta_k)$. Furthermore, we can express $\vec{m}_k \cdot \Gamma_k$ and $\vec{n}_k\cdot \Gamma_k$ in terms of rank-1 projectors. Formally, we have $\vec{m}_k\cdot \Gamma_k=\ket{\vec{m}_{k}}\!\bra{\vec{m}_{k}}-\ket{-\vec{m}_{k}}\!\bra{-\vec{m}_{k}}$ and similarly for $\vec{n}_k\cdot \Gamma_k$. Plugging these into (\ref{eq:effective.decomposition}),
\be \label{lem:c_k_theta}
c_k=\underset{i,j\in\{0,1\}}{\tn{max}} |\bra{(-1)^i\vec{m}_k}(-1)^j\vec{n}_k\rangle|^2 =\frac{1+\cos{\theta_k}}{2}
\ee
Next, we want to relate $c_k$ to the CHSH value. Using the result of Seevinck and Uffink~\cite{seevinck2007}, for all $r$, (\ref{CHSH:w_kr}) satisfies
\begin{align}
S_{k,r}\leq2\sqrt{1+\sin(\theta_k)\sin(\theta_r)} \label{CHSH:Seev}
\end{align}
where $\sin(\theta_k)$ and $\sin(\theta_r)$ quantify the commutativity of Alice's $k$th and system T's $r$th measurements, respectively. From (\ref{lem:c_k_theta}) and (\ref{CHSH:Seev}) we obtain for all $r$, 
\[c_k \leq \frac{1}{2}+\frac{S_{k,r}}{8}\sqrt{8-S_{k,r}^2} ,\]
where we use the fact that the right hand side is a monotonic decreasing function. Finally,  we get
\[c^*=\sum_k p_k c_k=\sum_{k,r}\lambda_{k,r}c_k \leq \frac{1}{2}+ \frac{S}{8}\sqrt{8-S^2},\]
and the inequality is given by the Jensen's inequality and (\ref{decom:CHSH}).

\end{proof}

\subsection{Statistics and efficiency of Charlie's operation}
We recall in the protocol description, after the sifting step, Alice and Bob identify sets $\cX, \cZ$ and $\cJ$. Also, they have $\tilde{\cX}$ where $|\tilde{\cX}|$ corresponds to the total number of times Alice chooses sub-protocol $\Gamma_{\tn{QKD}}$, and both Alice and Bob choose setting $\mathsf{X}$.

 Part of the goal is to estimate the average overlap of set $\cX$ with the observed CHSH values (evaluated on sets $\cJ$) and the efficiency of Charlie's operation $\eta$. Note that $\eta=|\cX|/|\hat{\cX}|$. To do that, we need the following two lemmas: the first (Lemma \ref{lem:statistics.cstar}) gives a bound on the average effective overlap of $\cX$ in terms of the average effective overlap of $\tilde{\cX}$ and the efficiency of Charlie's operation $\eta$, and the second (Lemma \ref{lem:statistics.chsh}) gives a bound on the probability that the observed CHSH value is larger than the expected CHSH value.

\begin{lem} \label{lem:statistics.cstar}
Let $c^*_\cX$ and $c^*_{\tilde{\cX}}$ be the
average effective overlaps of $\cX$ and $\tilde{\cX}$, respectively, and let $\eta:=|\cX|/|\tilde{\cX}|$. Then 
\[
c^*_\cX \leq \frac{1}{2}+ \frac{1}{\eta}\left(c^*_{\tilde{\cX}}-\frac{1}{2}	 \right)
\]
\end{lem}
\begin{proof}
First, we note that $\cX\subseteq \tilde{\cX}$ with equality only if Charlie always outputs a pass (or has perfect efficiency).  Next, we consider $\{c^{*,i}\}_{i\in \tilde{\cX}}$ in decreasing order, that is, $c^{*,1}\geq c^{*,2}\geq \cdots \geq c^{*,|\tilde{\cX}|}$. Accordingly, the average overlap of $\tilde{\cX}$ can be written as
\[c_{\tilde{\cX}}^*=\frac{|\cX|}{|\tilde{\cX}|}\sum_{i=1}^{|\cX|}\frac{c^{*,i}}{|\cX|}+\sum_{j=|\cX|+1}^{|\tilde{\cX}|}\frac{c^{*,i}}{|\tilde{\cX}|} \geq \frac{|\cX|}{|\tilde{\cX}|}\left(c^*_{\cX}-\frac{1}{2}\right)+ \frac{1}{2} \]
where we consider that $\cX$ collects the large effective overlaps, and the inequality is given by $c^{*,i}\geq 1/2$. Finally, let $\eta=|\cX|/|\tilde{\cX}|$.
\end{proof}

\begin{lem} \label{lem:statistics.chsh} Let $S_{\cJ}$ be the
  average CHSH value on $m_j$ independent systems, and
  $S_{\tn{test}}$ the observed CHSH on these systems. Then \[
  \Pr\left[S_{\tn{test}}-S_{\cJ} \geq \sqrt{\frac{32}{m_j}\ln
        \frac{1}{\eps}} \right] \leq \eps.\]
\end{lem}

\begin{proof}
We define the random variable \[Y_i
\coloneqq \begin{cases} 1 & \tn{if } s_i \oplus t_i = u_i \wedge u_i, \\
  0 & \tn{otherwise},\end{cases}\] where $u_i,v_i,s_i,t_i$ are the
inputs and outputs, respectively of the measurements on system $i$, and
$Y_\cJ \coloneqq \frac{1}{m_j} \sum_{i \in \cJ}Y_i$.  It is
easy to see that $S_i = 8\E[Y_i]-4$, $S_{\cJ} = 8\E[Y_\cJ]-4$
and $S_{\tn{test}} = Y_\cJ$. The proof is then immediate from Lemma \ref{lem:chernoff}.
\end{proof}

\subsection{Secrecy analysis}

With the relevant results in hand, we are ready to prove our main
result which roughly follows the same line of argument as
Ref.~\cite{Tomamichel2012a}. The main differences are the use of a
more general smooth entropic uncertainty relation (Lemma
\ref{thm:ucr}) to bound the error on the secrecy, and the CHSH test to
bound the effective overlap of the measurement operators and states
used by the uncertainty relation (Lemma \ref{lem:chsh_c}). Since the
players can only sample the CHSH violation, we use Lemma
\ref{lem:statistics.cstar} to bound the distance between this estimate
and the expected effective overlap of the key set. The correctness of
the protocol are evaluated in exactly the same way as in
Ref~\cite{Tomamichel2012a}, so we refer to that work for the
corresponding bounds and theorems. We only prove the secrecy of the
protocol here.

Contrary to most QKD protocols, the protocol adopts a tripartite model where Charlie is supposed to establish entanglement between Alice and Bob. Thus in our picture, we can view Charlie as an accomplice of the adversary and evaluate the secrecy on the overall state conditioned on the events where Charlie outputs a pass.

We briefly recall the main parameters of the protocol, which are
detailed in the protocol definition given in the paper. Conditioned on
the successful operation of Charlie (the events whereby Charlie
outputs a pass), Alice and Bob generate systems until at least $m_x$
of them have been measured by both of them in the basis $\mathsf{X}$,
$m_z$ have been measured in the basis $\mathsf{Z}$, and $j$ have been
chosen for both CHSH tests. The tolerated error rate and the CHSH
value are $Q_{\tn{tol}}$ and $S_{\tn{tol}}$, respectively.

Furthermore, we take that our information reconciliation scheme leaks
at most $\leak+\lceil\log(1/\epscorr)\rceil$-bits of information,
where an error correction scheme which leaks at most $\leak$-bits of
information is applied~\cite{RennerThesis}, then an error verification scheme using two-universal hashing which leaks
$\lceil\log(1/\epscorr)\rceil$-bits of information is
applied. If the error verification fails, they
abort the protocol.

\begin{thm}\label{securitythm}
  The protocol is $\eps_{\tn{sec}}$-secret if for some
  $\eps_Q,\eps_{\tn{UCR}},\eps_{\tn{PA}},\eps_{c^*},\eps_{\tn{CHSH}}
  > 0$ such that $4 \eps_Q + 2\eps_{\tn{UCR}} + \eps_{\tn{PA}} +
  \eps_{c^*} + \eps_{\tn{CHSH}} \leq \eps_{\tn{sec}}$, the final
  secret key length $\ell$ satisfies
  \begin{multline}\label{eq:securitythm}
  \! \!\! \ell \leq m_x \bigg( 1 - \log_2
    \bigg(1+\frac{\hat{S}_{\tn{tol}}}{4\eta_{\tn{tol}}} \sqrt{8 -
      \hat{S}_{\tn{tol}}^{2}} + \zeta(\eps_{c^*})\bigg)  - \tn{h}(\hat{Q}_{\tn{tol}})
    \bigg) \\ - \tn{leak}_{\tn{EC}} - \log_2
    \frac{1}{\eps_{\tn{UCR}}^2\eps_{\tn{PA}}^2\epscorr} ,
    \end{multline}
  where $\hat{S}_{\tn{tol}} :=S_{\tn{tol}} - \xi(\eps_{\tn{CHSH}})$ and $ \hat{Q}_{\tn{tol}} :=
   Q_{\tn{tol}}  + \mu(\eps_Q)$ with the statistical deviations given by 
\[ 
    \xi(\eps_{\tn{CHSH}})  \coloneqq \sqrt{\frac{32}{m_j}\ln \frac{1}{\eps_{\tn{CHSH}}}},\]
  \[  \zeta(\eps_{c^*}) \coloneqq \sqrt{\frac{2(m_x+m_j\eta)(m_j+1)}{m_xm_j^2}\ln \frac{1}{\eps_{c^*}}},\\\] and \[ 
    \mu(\eps_Q)  \coloneqq \sqrt{\frac{(m_x+m_z)(m_z+1)}{m_x m_z^2} \ln
      \frac{1}{\eps_Q}}.
\]
\end{thm}

\begin{proof}
Let $\Omega$ be the event that $Q_{\tn{test}} \leq Q_{\tn{tol}}$ and $S_{\tn{test}} \geq S_{\tn{tol}}$ and $\eta \geq \eta_{\tn{tol}}$.  If $\Omega$ fails to occur, then the protocol aborts, and the secrecy error is trivially zero. Conditioned on passing these tests, let $X$ be the random variable on strings of length $m_x$ that Alice gets from the set $\cX$, and let $E$ denote the adversary's information obtained by eavesdropping on the quantum channel. After listening to
  the error correction and hash value, Eve has a new system
  $E'$. Using $\lceil\log(1/\epscorr)\rceil \leq \log_2(2/\epscorr)$ (the number bits used for error correction and error verification) and using chain rules for smooth entropies~\cite{RennerThesis} we bound the min-entropy of the $X$ given $E'$
\[\HminSmooth{2\eps+\eps_{\tn{UCR}}}{X|E'} \geq
\HminSmooth{2\eps+\eps_{\tn{UCR}}}{X|E} - \tn{leak}_{\tn{EC}}-\log_2\frac{2}{\epscorr}.\]
From the entropic uncertainty relation (Lemma \ref{thm:ucr}), we further get
\[\HminSmooth{2\eps+\eps_{\tn{UCR}}}{X|E} \geq \log_2 \frac{1}{c^*} -
\HmaxSmooth{\eps}{Z|B} - \log_2 \frac{2}{\eps_{\tn{UCR}}^2},\] where $Z$ can
be seen as the outcome Alice would have gotten if she had measured the
same systems in the corresponding basis $\mathsf{Z}$, and $B$ is Bob's
system in this case (before measurement). 

The max-entropy of the alternative measurement is then
bounded by the error rate sampled on the $m_z$
systems $\cZ$~\cite{Tomamichel2012a}:
\[\HmaxSmooth{\eps}{Z|B} \leq m_x \tn{h}(Q_{\tn{tol}} +
\mu(\eps_Q)),\] 
where $\eps = \eps_Q/\sqrt{p_{\Omega}}$
and $p_{\Omega}:=\Pr[\mathsf{\Omega}]$.

Next, we bound $c^*$ (evaluated on Alice's devices from $\cX$) in terms of the observed CHSH value
$S_{\tn{test}}$. We first use the arithmetic-geometric mean's inequality, from which we get
\[ c^*\leq\prod_{i \in \cX} c^{*,i} \leq \left(\sum_{i \in \cX}
  \frac{c^{*,i}}{m_x}\right)^{m_x} = \left(c^*_\cX\right)^{m_x},\] where $c^*_{\cX}$ is the average effective overlap on $\cX$. Using Lemma. \ref{lem:statistics.cstar}, we get
\[
c^*_{\cX} \leq \frac{1}{2}+\frac{1}{\eta}\left(c_{\tilde{\cX}}^*-\frac{1}{2} 	\right).\] Since $\tilde{\cX}$ is randomly chosen by Alice and is independent of Charlie, $c^*_{\tilde{\cX}}$ can be estimated from $c^*_{\cJ}$, i.e., we apply Corollary \ref{Cor:serfling} to further obtain
$\Pr [ c^*_{\tilde{\cX}} - c^*_{\cJ} \geq \zeta (\eps_{c^*})/2 ]
\leq \eps_{c^*}$, hence \[ \eps' \coloneqq \Pr \left[ c^*_{\tilde{\cX}} -
  c^*_{\cJ} \geq \frac{\zeta (\eps_{c^*})}{2} \middle| \Omega
\right] \leq \frac{\eps_{c^*}}{p_{\Omega}}.\] 
Lemma \ref{lem:chsh_c} can now be used together with
Jensen's inequality, so with probability at least $1-\eps'$, \[c^*_{\tilde{\cX}}
\leq \frac{1}{2}\left( 1+\frac{S_{\cJ}}{4} \sqrt{8 -
    S_{\cJ}^2} + \zeta(\eps_{c^*})\right).\]  We still need to
take into account that we only have an approximation for the CHSH
value of the systems in $\cJ$. From Lemma \ref{lem:statistics.chsh} we
get that \[\eps'' \coloneqq \Pr \left[
  S_{\cJ}\leq\hat{S}_{\tn{test}}
  \middle|\Omega\right] \leq
\frac{\eps_{\tn{CHSH}}}{p_{\Omega}}.\]

Finally, the bound on the error of privacy amplification by
universal hashing~\cite{RennerThesis} says that the error is less than
$4\eps + 2\eps_{\tn{UCR}} + \eps_{\tn{PA}}$ as long as \[ \ell \leq
\HminSmooth{2\eps+\eps_{\tn{UCR}}}{X|E'} - 2 \log_2 \frac{1}{2\eps_{\tn{PA}}}.\]

Putting all the above equations together we get
(\ref{eq:securitythm}), with a total error conditioned on the event $\Omega$ of at most $4\eps + 2\eps_{\tn{UCR}} + \eps_{\tn{PA}}
+ \eps'+\eps''$. If we
remove this conditioning, the error is then
\begin{multline}\notag p_{\Omega} (4\eps + 2\eps_{\tn{UCR}} + \eps_{\tn{PA}}
+\eps'+\eps'') \\ \leq 4\eps_Q + 2\eps_{\tn{UCR}} + \eps_{\tn{PA}} +
\eps_{c*}+\eps_{\tn{CHSH}}\leq \esec. \qedhere \end{multline}
\end{proof}

\bibliographystyle{unsrt}

\begin{thebibliography}{99}
\bibitem{footnote1}
We label the measurements $\mathsf{Z}$ and $\mathsf{X}$ because measurements in the computational and diagonal basis are optimal. We stress however that our security proof is independent of the actual measurements that are carried out.
\bibitem{footnote2} The values of the probabilities $p_x$ and $p_s$
  are chosen such that they minimize the number of iterations of the
  protocol.
\end{thebibliography}

\end{document}